\documentclass[prd,preprint, nofootinbib,amsmath,amssymb, aps, floatfix]{revtex4-1}
\pdfoutput=1
\usepackage{dcolumn}
\usepackage{bm}
\usepackage{amssymb,amsmath,graphicx}
\usepackage[linktoc=page]{hyperref}
\usepackage{tikz}
\usepackage{fancyhdr}
\usetikzlibrary{decorations.pathmorphing, patterns}
\usepackage{pgfplots}

\usepackage{amsthm}
\newtheorem{theorem}{Theorem}
\newtheorem{lemma}[theorem]{Lemma}

\theoremstyle{definition}

\newcommand{\h}{\mathcal{H}}
\newcommand{\A}{\mathcal{A}}
\newcommand{\M}{\mathcal{M}}
\newcommand{\N}{\tilde{N}}

\newcommand{\bra}[1]{\langle #1|}
\newcommand{\ket}[1]{|#1\rangle}
\newcommand{\braket}[2]{\langle #1|#2\rangle}

\begin{document}

\title{Type III von Neumann Algebras are Magical}
\author{Mudassir Moosa}%
\email{mudassir.moosa20@gmail.com}
\affiliation{Department of Physics and\\ Center for Theory of Quantum Matter\\
University of Colorado\\ Boulder, CO 80309}

\begin{abstract}
The number of non-Clifford gates needed to perform a task, or simply \textit{magic}, is a resource for fault-tolerant quantum computation. Von Neumann algebras provide a formal mathematical structure to describe infinite-dimensional quantum systems, such as those in quantum field theory or quantum statistical mechanics. A particularly important class of von Neumann algebras is called Type III algebras, for which the standard notions of finite-dimensional systems such as density matrices and traces break down. In this work, we argue that Type III von Neumann algebras fundamentally require an infinite amount of magic. Specifically, we consider the thermodynamic limit of finite-dimensional quantum systems, such as lattice systems, and show that if the states in the thermodynamic limit possess only a bounded amount of magic, the resulting local von Neumann algebra cannot be of Type III. Our result has direct implications for the quantum simulations of quantum field theories, for which the algebra of a local subregion is known to be of Type III.
\end{abstract}

\maketitle

\section{Introduction} \label{intro}

Quantum simulations of strongly coupled quantum field theories and quantum gravity models are expected to provide novel insights into the physics of black holes and the emergence of spacetime \cite{Preskill:2018fag, Bauer:2022hpo}. Understanding the resource requirements to perform these simulations is thus a fundamental problem. In fault-tolerant quantum computing, the number of non-Clifford gates needed to perform a task is usually a suitable metric for the required resources for two main reasons. The first is due to the Gottesman-Knill theorem, which states that a quantum circuit involving only Clifford gates can be simulated on a classical computer in polynomial time \cite{Gottesman:1998hu}. The second reason is that the fault-tolerant implementation of non-Clifford gates requires the preparation and injection of certain \textit{magic states} \cite{Bravyi:2004isx}. The distillation of these magic states is highly resource-intensive and typically dominates the overhead of the overall algorithm. Therefore, quantifying the number of non-Clifford gates, or simply \textit{magic}, required to simulate strongly coupled quantum systems is an important problem. 

This line of study was initiated in Ref.~\cite{White:2020zoz}, where it was shown that the magic in the ground states of the $\mathbb{Z}_3$ Potts model at its critical point scales extensively with the system size. The low-energy physics of this spin chain is known to be described by a conformal field theory \cite{Cardy:1986ie,vonGehlen:1986gk}. This led to the conclusion that the preparation of a ground state of a conformal field theory on a quantum computer, which is the first step in running a quantum simulation, requires a large number of non-Clifford gates. 

More recently, Refs.~\cite{Cao:2023mzo,backreaction-magical, Cao:2026uoq} have explored the role of magic in the AdS/CFT correspondence and holographic error-correcting codes \cite{maldacena,Pastawski:2015qua}. In particular, these works have shown that if the holographic code is an exact stabilizer code \cite{Gottesman:1997zz} or an exact subsystem erasure-correcting code \cite{Pastawski:2016qrs}, it cannot support non-trivial area operators \cite{Harlow:2016vwg}. Consequently, if the state of the boundary theory does not possess magic, the bulk matter fields cannot backreact on the gravitational spacetime, and hence, the bulk spacetime cannot be dynamic.

Our goal in this work is to extend these results using the formalism of von Neumann algebras, which play a crucial role in the study of infinite-dimensional quantum systems. To illustrate this, consider an infinite-dimensional system that arises as a thermodynamic limit of a finite quantum system. As we take the thermodynamic limit, the Hilbert space splits into disjoint superselection sectors corresponding to different macroscopic phases \cite{Haag:1996hvx}. The states within a specific sector have macroscopic properties (e.g., non-zero mean magnetization or charge) and exhibit large entanglement between the degrees of freedom inside a local region and those outside it. Due to this infinite entanglement, the notion of a local subsystem defined by the tensor factorization of the Hilbert space ceases to be meaningful. To resolve this, the theory of operator algebras provides a formal framework where a local subsystem corresponding to a subregion is defined by the von Neumann algebra of bounded operators localized to that spatial region.

Furthermore, the specific structure of the von Neumann algebra associated with a local subregion depends on the superselection sector in which it is represented. This is because the macroscopic properties of states in a given superselection sector, such as mean magnetization or temperature, dictate the convergence of sequences of local operators. A particularly important class of von Neumann algebras consists of those for which the standard notions of finite-dimensional systems, such as density matrices and traces, break down \cite{Sorce:2023fdx}. These are known as Type III von Neumann algebras (which we formally define in Sec.~\ref{sec-algebras}), and they are precisely the type of von Neumann algebras that are associated with local subregions in quantum field theories \cite{Haag:1996hvx}. 

\begin{figure}
    \centering
    \includegraphics[width=0.67\linewidth]{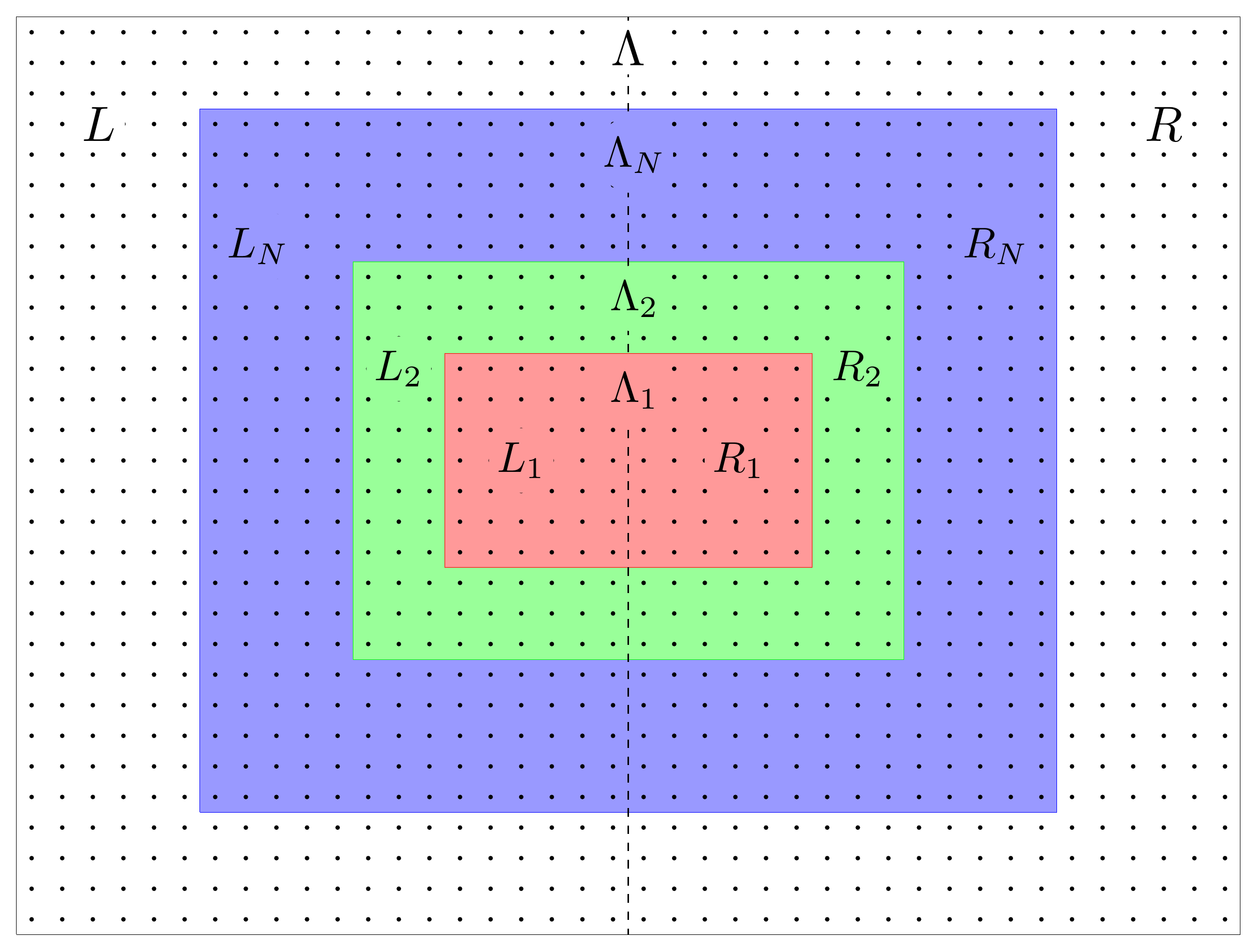}
    \caption{A schematic representation of a nested sequence of lattices $\Lambda_1 \subset \Lambda_2 \subset \cdots$ that limits to an infinitely large lattice $\Lambda$. Each site (denoted by a black dot) contains a single qubit. Thus, the sequence of lattices defines a nested sequence of Hilbert spaces $\h_1 \subset \h_2 \subset \cdots$. Furthermore, each finite lattice in the sequence is partitioned into a right and a left sublattice, i.e., $\Lambda_N = R_N \cup L_N$ for finite $N$, which define subalgebras $\A_N \subset B(\h_N)$ of operators acting non-trivially only on qubits in $R_N$. Note that the lattices are shown to be two-dimensional and symmetric only for the sake of simplicity, though our results are applicable to higher-dimensional and non-symmetric lattices as well.}
    \label{fig-lattice}
\end{figure}

In this work, our goal is to show that the von Neumann algebra can only be of Type III if the superselection sector contains states with infinite magic. To do this, we consider infinite-dimensional quantum systems that arise as the thermodynamic limit of lattice models. More precisely, we consider a nested sequence of finite lattices, $\Lambda_1 \subset \Lambda_2 \subset \cdots$, where the limiting lattice is infinitely large, as shown in Fig.~\ref{fig-lattice}. We assume that each site $x\in\Lambda_N$ contains a single qubit\footnote{This assumption is made purely for the sake of concreteness and our results are also valid if we consider a $d$-dimensional qudit at each site.}, and hence, we get a nested sequence of Hilbert spaces $\h_1 \subset \h_2 \subset \cdots$, where $\h_N = \otimes_{x\in\Lambda_N} \mathbb{C}^2_{x}$. Once we specify how smaller Hilbert spaces are embedded into larger Hilbert spaces in this sequence, we can take the thermodynamic limit to define an infinite-dimensional separable Hilbert space $\h$. Note that a different choice of embedding leads to a different Hilbert space with different macroscopic properties (i.e., a different superselection sector) in the thermodynamic limit. (More precisely, the sequence of Hilbert spaces along with the embedding maps defines an \textit{inductive system}, and $\h$ is the \textit{inductive limit} of this system; see Sec.~\ref{sec-inductive-limit} for a review.)

Moreover, we partition each lattice $\Lambda_N$ into a right sublattice (denoted by $R_N$) and a left sublattice (denoted by $L_N$), such that $R_N \subset R_M$ and $L_N \subset L_M$ for $M>N$, as shown in Fig.~\ref{fig-lattice}. If we denote the algebra of operators that act trivially on $L_N$ by $\A_N \subset B(\h_N)$, we obtain a sequence of nested algebras $\A_1 \subset\A_2 \subset \cdots$. Taking the thermodynamic limit (inductive limit) yields an infinite-dimensional von Neumann algebra $\A \subset B(\h)$ corresponding to the observables localized in the right region.

As noted earlier, the resulting von Neumann algebra $\A \subset B(\h)$ depends on the Hilbert space $\h$ (i.e., the superselection sector). In particular, we prove that if the states in $\h$ contain only a bounded amount of magic, then $\A$ cannot be of Type III. Therefore, our result establishes that Type III von Neumann algebras require infinite magic.

The rest of this paper is organized as follows. We present the necessary background in Sec.~\ref{background}. Specifically, we briefly review stabilizer states and discuss a few measures of magic in Sec.~\ref{stabilizer states}. In Sec.~\ref{sec-algebras}, we present a review of von Neumann algebras, focusing on the properties of projections and linear functionals on von Neumann algebras. Then, in Sec.~\ref{sec-inductive-limit}, we discuss the formalism of inductive systems and review the construction of the inductive limit Hilbert space and algebra. We prove our main result in Sec.~\ref{sec-main-thms}. Finally, we end with a discussion of the implications and extensions of our work in Sec.~\ref{sec-discussion}. 

{\bf Note: } While we were finalizing this manuscript, Ref.~\cite{Benedetti:2026mfy} appeared on arXiv, which independently demonstrated that the vacuum states of Lorentz-invariant QFTs require non-zero magic. Our work complements this result by establishing that Type III algebras fundamentally require unbounded magic. Furthermore, our approach proves this without relying on specific assumptions, such as Lorentz invariance or the Reeh-Schlieder theorem.

\section{Background and Preliminaries} \label{background}

\subsection{Stabilizer States and Magic} \label{stabilizer states}

Here, we review the theory of stabilizer states and magic for finite-dimensional Hilbert spaces. For simplicity, we only consider a system of $N$ qubits, though the discussion can be generalized to $N$ qudits as well.

Consider a Hilbert space of $N$ qubits. An operator acting on this space is called a Pauli string if it is written as the product of Pauli or identity operators, $\{X,Y,Z,I\}$, at each qubit. These $4^N$ Pauli strings, along with the phase factors $\{\pm1 , \pm i\}$, form a group called the Pauli group, which we denote by $\mathcal{P}_N$. 

A unitary operator is called a \textit{Clifford} if it maps an element of the Pauli group to another element of the Pauli group. For example, a unitary that can be written as a product of single-qubit Hadamard operators $H = (X+Z)/\sqrt{2}$, single-qubit $\pi/2$-rotation operators $R_z(\pi/2) = e^{-i\pi Z/4}$, and two-qubit entangling operators $e^{i\pi Z\otimes Z / 4}$ is a Clifford unitary. The Clifford unitaries form a group called the \textit{Clifford group}, which we denote by $\mathcal{C}_N$. In other words, the Clifford group $\mathcal{C}_N$ is the normalizer of the Pauli group in the group of all $N$-qubit unitary operators, that is, $\mathcal{C}_N^\dagger \mathcal{P}_N \mathcal{C}_N = \mathcal{P}_N \, $.

Let us now consider a quantum state of the form $\ket{\psi} \, = \, {C}\ket{0}^{\otimes N}$, where $C\in \mathcal{C}_N$. It is straightforward to show that this state is an eigenstate of $2^N$ distinct Pauli strings with a $+1$ eigenvalue. These $2^N$ Pauli strings are of the form $C (Z^{a_1}\cdots Z^{a_N}) C^\dagger$ for $a_{i} \in \{0, 1\}$, and hence, commute with each other. The state $\ket{\psi}$ is called a \textit{stabilizer state} and the $2^N$ Pauli strings discussed above are called the stabilizers of this state. 

Any abelian subgroup of the Pauli group $\mathcal{P}_N$ that does not contain the $-I$ operator is called a \textit{stabilizer group}. Therefore, the stabilizers of a stabilizer state $\ket{\psi}$ form a stabilizer group, which we denote by $\mathcal{S}_{\psi}$. Conversely, for any stabilizer group that contains $2^N$ elements, there is a unique pure state which is an  eigenstate with a $+1$ eigenvalue for every element of that group.

An interesting property of stabilizer states is that they only allow for a flat spectrum. More precisely, for a stabilizer state $\ket{\psi}_{AB}$ of a bipartite system, the reduced states $\rho_A$ and $\rho_B$ are proportional to projection operators. This property simplifies the calculation of the entanglement entropy in stabilizer states. In this work, this property will play a crucial role in the proofs of our main results in Sec.~\ref{sec-main-thms}, and hence, we present this as the following lemma.
\begin{lemma}[\cite{Fattal:2004frh}, Result $1$]  \label{lemma-stab-proj}
    Consider a pure stabilizer state $\ket{\psi}$ of $N$ qubits with a stabilizer group $\mathcal{S}_\psi$. Let us split the $N$ qubits into two disjoint sets $A$ and $B$ consisting of $N_A$ and $N_B$ qubits, respectively. The reduced state for the $A$ set of qubits is given by 
    \begin{align}
        \rho_A \, = \, \frac{1}{2^{N_A}} \, \sum_{\mathcal{P} \in \mathcal{S}_{\psi ; A}} \mathcal{P} \,  ,
    \end{align}
    where $\mathcal{S}_{\psi ; A} \subset \mathcal{S}_\psi$ is the set of all stabilizers of $\ket{\psi}$ which act as $I$ on the set $B$. Moreover, $\rho_A^2 = 2^{-N_A} \rho_A$, and hence, $\rho_A$ is proportional to a projection operator. 
\end{lemma}
\noindent We refer the readers to Ref.~\cite{Fattal:2004frh} for a proof of this lemma. 

One implication of this lemma is that the Clifford group does not form a universal set of quantum gates since Clifford unitaries acting on $\ket{00\cdots 0}$ only prepare states with a flat spectrum. Therefore, universal quantum computation requires access to non-Clifford gates (like Toffoli or T gates). Moreover, in fault-tolerant protocols, non-Clifford gates are typically implemented through the injection of magic states which are resource-intensive to prepare \cite{Bravyi:2004isx}. Therefore, the amount of \textit{magic} or the non-stabilizerness of a quantum state, which roughly speaking measures the number of non-Clifford gates required to prepare the state, is a fundamental resource for universal fault-tolerant quantum computation \cite{Emerson:2013zse,Howard:2017maw}.

There are many measures to quantify the magic in a pure state, such as \textit{mana} \cite{Bravyi:2004isx}, \textit{stabilizer R\'enyi entropy} \cite{Leone:2021rzd}, or \textit{stabilizer fidelity} \cite{Bravyi:2018ugg,Liu:2020yso}. In particular, the stabilizer fidelity, $F_{\text{stab}}(\ket{\psi})$, is defined as \cite{Bravyi:2018ugg}
\begin{align}
    F_{\text{stab}}(\ket{\psi}) \, = \, \max_{\ket{\phi}\in \text{STAB}_N} \vert\braket{\psi}{\phi}\vert^{2}  \, ,
\end{align}
where $\text{STAB}_N$ is the set of all $N$-qubit stabilizer states. A closely related quantity is the \textit{min-relative entropy of magic}, which is defined as \cite{Liu:2020yso}
\begin{align}
\M_{\text{stab}}(\ket{\psi}) \, = \, - \log F_{\text{stab}}(\ket{\psi}) \, .    
\end{align}
It is easy to see that $\M_{\text{stab}}(\ket{\psi}) = 0$ if and only if $\ket{\psi}$ is a stabilizer state. Note that when the number of qubits $N$ is finite, $\M_{\text{stab}}$ is bounded from above. However, $\M_{\text{stab}}$ is unbounded in the thermodynamic limit ($N \to\infty$).

\subsection{Von Neumann Algebras} \label{sec-algebras}

Consider a $*$-subalgebra of bounded operators on a Hilbert space, $\mathcal{A} \subset B(\mathcal{H})$, that contains the identity operator. A sequence of operators $\{a_n\} \subset \mathcal{A}$ is said to converge to $a \in \A$ in the weak operator topology (WOT) if 
\begin{align}
    \text{WOT : }\quad \lim_{n\to\infty} |\bra{\psi}(a_n-a)\ket{\phi}| = 0
\end{align}
for all $\ket{\phi}, \ket{\psi} \in \mathcal{H}$. The closure of $\mathcal{A} \subset B(\mathcal{H})$ in the WOT defines a von Neumann algebra. \cite{Haag:1996hvx,Takesaki}. 

Another useful topology is the strong operator topology (SOT). A sequence $\{a_n\} \subset \mathcal{A}$ is said to converge to $a\in\A$ in the SOT if
\begin{align}
    \text{SOT : } \quad \lim_{n\to\infty} ||(a_n-a)\ket{\phi}|| = 0
\end{align}
for all $\ket{\phi} \in \mathcal{H}$. The SOT is stronger than the WOT in the sense that convergence in the SOT implies convergence in the WOT. Remarkably, von Neumann's bicommutant theorem asserts that the SOT and WOT closures of $\mathcal{A}$ are equal and coincide with its double commutant. That is, the von Neumann algebra $\mathcal{A}$ satisfies 
\begin{align}
    \mathcal{A}  = \overline{\mathcal{A}}^{\text{WOT}} = \overline{\mathcal{A}}^{\text{SOT}} = \mathcal{A}'' \, ,
\end{align}
where $\mathcal{A}' = \{ a' \in B(\mathcal{H}) | [a' , a] = 0  \, \forall a \in \mathcal{A}\}$ is the commutant of $\mathcal{A}$ \cite{Haag:1996hvx,Takesaki}.

A simple example of a von Neumann algebra is the algebra of operators acting on a subsystem of a bipartite system with a finite-dimensional Hilbert space $\mathcal{H} = \mathcal{H}_A \otimes \mathcal{H}_B$. In particular, $\A = B(\mathcal{H}_A)\otimes \mathbf{1}_B$ is a von Neumann algebra and $\A' = \mathbf{1}_A\otimes B(\mathcal{H}_B)$ is its commutant.

In the above example of finite-dimensional Hilbert spaces, the von Neumann algebras are isomorphic to algebras of finite-dimensional matrices. These von Neumann algebras are said to be of Type I$_d$, where $d$ denotes the dimension of the underlying Hilbert space. In contrast, in infinite dimensions, not all von Neumann algebras acting on a Hilbert space are isomorphic, and they are classified based on the types of projection operators that they admit.

\subsubsection{Projections}

An operator $P \in B(\mathcal{H})$ is a \textit{projection} if $P = P^\dagger = P^2 \, $. Here, we collect and review some properties of projections that we make use of in proving our main results in Sec.~\ref{sec-main-thms}. 

A projection $P\in \A$ defines a von Neumann subalgebra. That is, if $\A \subset B(\h)$ is a von Neumann algebra, then $P\mathcal{A}P \subset B(P\mathcal{H})$ is also a von Neumann algebra \cite{Takesaki}. 
In finite-dimensional algebras, the `size' of a projection $P\in\mathcal{A}$ can be defined in terms of the dimension of the reduced subalgebra $P\mathcal{A}P$. However, in infinite-dimensional algebras, the reduced subalgebra $P\mathcal{A}P$ may also be infinite-dimensional, and hence, there is no absolute notion of the size of a projection. Nevertheless, we can define the relative sizes of two projections and use them to define a notion of a \textit{finite} projection. First, we say that two projections satisfy $P_1 \ge P_2$ if $P_1 P_2 = P_2 P_1 = P_2$, or equivalently if $P_1 - P_2$ is a projection. In other words, $P_1 \ge P_2$ if the image of $P_2$ is a subspace of the image of $P_1$. Secondly, we say that two projections are equivalent, or of the same size, if their subspaces can be isometrically mapped to each other. More precisely, two projections $P, Q \in \mathcal{A}$ are equivalent, denoted by $P \sim Q$, if there exists a partial isometry $V \in \mathcal{A}$ such that $V^\dagger V = P$ and $V V^\dagger = Q$. Finally, a projection $P$ is \textit{finite} if there is no other projection $Q$ satisfying $Q<P$ and $Q\sim P$.

The von Neumann algebras that do not contain non-zero finite projections (i.e., every projection is either zero or infinite) are said to be of Type III. Conversely, if a von Neumann algebra contains a non-zero finite projection, it cannot be of Type III \footnote{Von Neumann algebras that contain finite projections can be further classified into Type I or Type II based on whether they contain \textit{minimal/abelian} projections. This subclassification is not relevant for this work, and hence, we refer interested readers to Refs.~\cite{Sorce:2023fdx,Takesaki,BratteliRobinson1987} for a review.}. We will make use of this criterion in Sec.~\ref{sec-main-thms} to show how infinite magic is a necessary condition for Type III algebras. 

\subsubsection{Functionals on von Neumann algebras}

A functional on a von Neumann algebra $\A$ is a linear map $\phi: \A \to \mathbb{C}$. 
A functional $\phi$ on $\A$ is \textit{positive} if  $\phi(a^\dagger a) \ge 0$. Furthermore, a functional is called a \textit{state functional} if it is positive and normalized, i.e., $\Vert \phi \Vert \equiv \phi(I) = 1$. The proof of our main theorem in Sec.~\ref{sec-main-thms} will make use of properties of functionals on von Neumann algebras and the topology on the space of these functionals. Therefore, we briefly review the important properties that we will need later. 

In a finite-dimensional von Neumann algebra $\A$, any positive (state) functional can be written as $\phi(\cdot) = \text{tr}(\rho \, \cdot)$ for some (normalized) density matrix $\rho$. This state functional is continuous in the sense that if there is a sequence of operators $\{O_m \in \A\}$ that converges to $O\in \A$, then $\lim_{m}\phi(O_m) = \phi(O)$. However, in the infinite-dimensional case, not all functionals are continuous \footnote{\label{example-singular-functional} For example, consider a one-dimensional quantum particle and a normalized wavefunction $\phi_K(x)$ that has support in a small neighborhood $x \in (K-\epsilon,K+\epsilon)$ and a limiting state functional $\phi(\cdot) \equiv \lim_{K\to\infty} \bra{\phi_K}\cdot\ket{\phi_K}$. Clearly, $\phi(P_m) = 0$ for all finite $m$, where $P_m$ is a projection that detects whether a particle is in $|x|<m$. Therefore, $\lim_{m} \phi(P_m) = 0 \ne \phi(\lim_{m}P_m) = \phi(1)=1$.}. The functionals that satisfy $\lim_{m} \phi(O_m) = \phi(\lim_m O_m)$ for any increasing sequence $O_m \in \A$ are called \textit{normal} functionals. Moreover, by the non-commutative version of Lebesgue decomposition, any general positive bounded functional $\phi$ can be uniquely decomposed into a normal part $\phi^{(n)}$ and a non-normal (or singular) part $\phi^{(s)}$ as 
\begin{align}
    \phi = \phi^{(n)} + \phi^{(s)} \, , \label{functional-decomposition}
\end{align}
where both $\phi^{(n)}$ and $\phi^{(s)}$ are positive functionals and where $\Vert \phi\Vert = \Vert \phi^{(n)} \Vert + \Vert\phi^{(s)}\Vert$ \cite{takesaki1958conjugate} (also see \cite[Thm.~10.1.15(iii)]{kadison-ringrose-2}).

Another class of functionals is \textit{faithful} functionals for which $\phi(O^\dagger O) = 0$ implies $O =0$. (In finite dimensions, a functional is faithful if the corresponding density matrix is of full rank.) It is known that for any normal functional, there exists a subalgebra on which it is faithful. More precisely, for any normal positive functional $\phi$ on $\A$, there exists a non-zero (but not necessarily finite) \textit{support projection}, which we denote by $P_{\phi}$, such that $\phi$ is faithful on the subalgebra $P_\phi \A P_\phi$ and $\phi(\cdot) = \phi(P_\phi \, \cdot) = \phi(\cdot \, P_\phi)$ \cite{Takesaki}. 

An important example of a functional on a von Neumann algebra is a \textit{trace}. A trace is a normal and faithful functional on $\A$ such that $\tau(O' O) = \tau(O O')$ for all $O,O'\in\A$. However, a non-trivial (i.e., semifinite) trace is not well-defined for all von Neumann algebras. In particular, a von Neumann factor admits a non-trivial trace if and only if the factor contains finite projections \cite{Takesaki}. The following lemma further establishes the connection between finite projections and traces, and it will play a crucial role in the proof of our main results in Sec.~\ref{sec-main-thms}.
\begin{lemma}\label{lemma-tracial-corner}
    Let $\mathcal{A} \subset B(\mathcal{H})$ be a von Neumann algebra and let $P \in \mathcal{A}$ be a projection. If the subalgebra $P\A P \subset B(P\h)$ admits a trace $\tau$ such that $\tau(P) < \infty$, then $P$ is a finite projection in $\mathcal{A}$.
\end{lemma}
\begin{proof}
    See Appendix~\ref{appendix-proof-trace}.
\end{proof}

Finally, we discuss a useful topology on the space of state functionals on $\A$, called the weak-$^*$ topology. A sequence of functionals $\phi_n$ is said to converge to $\phi$ in the weak-$^*$ topology if 
\begin{align}
    \text{Weak-$^*$ : } \lim_{n\to\infty} \phi_n(O) = \phi(O) \, ,
\end{align}
for all $O \in \A$. An elementary result in functional analysis, called the Banach-Alaoglu theorem \cite{BratteliRobinson1987}, asserts that the space of state functionals on $\A$ is compact in this topology. This means that if we have a sequence of state functionals $\phi_n$, there exists a subsequence that converges in the weak-$^*$ topology to some state functional.

\subsection{Inductive Limits} \label{sec-inductive-limit}

As discussed in Sec.~\ref{intro}, we are interested in the infinite-dimensional Hilbert spaces and von Neumann algebras that arise from the thermodynamic or continuum limits of lattice models. A formal method to describe these limits is through the framework of inductive systems \cite{Chemissany:2025vye,kadison-ringrose-2}, which we briefly review here. 

Suppose we have an infinite sequence of nested Hilbert spaces $\h_1 \subset \h_2 \subset \cdots \, $ and let $(\cdot, \cdot)_{N}$ denote the inner product on $\h_N$. We introduce an embedding map $\iota_{N,M} : \h_N \to \h_M$ for $M > N$ that specifies how a smaller Hilbert space $\h_N$ is embedded inside a bigger space $\h_M$. This map is compatible with the inner product, i.e., $\left( \iota_{N,M}(\ket{\phi}_N), \iota_{N,M}(\ket{\psi}_N)\right)_{M} = \left( \ket{\phi}_N, \ket{\psi}_N\right)_{N}$ for all $\ket{\phi}_N, \ket{\psi}_N \in \h_N$, and it satisfies the consistency condition $\iota_{N,L} = \iota_{M,L}\circ\iota_{N,M}$ for $N<M<L$.

The infinite sequence and the embedding maps together define an inductive system, which leads to an inductive limit Hilbert space as follows. We first define a Hilbert space $[\h_N]$ which consists of the vectors $[\ket{\psi}_N] \equiv \lim_{M\to\infty}\iota_{N,M}(\ket{\psi}_N)$ for all $\ket{\psi}_N\in\h_N $. Then we define an inner product space $\h_{\text{union}} \equiv \cup_{j=1}^{\infty} [\h_{j}]$ with an inner product $([\ket{\psi}_N] , [\ket{\phi}_M]) \, = \, \left(\iota_{N,M}(\ket{\psi}_N) , \ket{\phi}_M\right)_M $ if $M\ge N$ and $([\ket{\psi}_N] , [\ket{\phi}_M]) \, = \, \left(\ket{\psi}_N , \iota_{M,N}(\ket{\phi}_M)\right)_N $ if $M< N$. Completing this space with the norm induced by the inner product defines a Hilbert space, which we denote by $\h$. This Hilbert space is the inductive (or direct) limit of the infinite sequence $\h_1\subset\h_2\subset\cdots$ with the embedding maps $\iota_{N,M}$ \cite{Chemissany:2025vye,kadison-ringrose-2}.

Note that the inductive limit Hilbert space $\h$ depends on both the sequence of Hilbert spaces and the embedding maps. As an example, suppose $\h_N$ is a Hilbert space of $2N$ qubits and suppose the embedding maps $\iota_{N,N+1}$ append a fixed entangled state $\ket{\lambda} = \sqrt{\lambda}\ket{00}+\sqrt{1-\lambda}\ket{11}$ to the additional $2$ qubits, i.e., $\iota_{N,N+1}(\ket{\psi}_N) = \ket{\psi}_N \otimes \ket{\lambda}$ for all $\ket{\psi}_N \in \h_N$. As a result, a dense subspace of $\h$ consists of states that have a tail resembling an infinite tensor product of  $\ket{\lambda}$.

Next, we discuss the algebra of operators acting on $\h$. Suppose we have a sequence of nested algebras $\A_{1} \subset \A_2 \subset \cdots$, where $\A_N \subset B(\h_N)$. Moreover, suppose we have embedding maps $\gamma_{N,M}: \A_N \to \A_M$ which are $*$-homomorphisms and are compatible with the embedding maps $\iota_{N,M}$, that is,
\begin{align}
    \gamma_{N,M}(O_N) \iota_{N,M}(\ket{\psi}_N) \, = \, \iota_{N,M}\left( O_N\ket{\psi}_N \right) \, , \label{eq-map-compatibility}
\end{align}
for all $O_N \in \A_N$ and $\ket{\psi}_N \in \h_N$. Since these embedding maps are $*$-homomorphisms, it follows that if $P_N$ and $Q_N$ are projections in $\A_N$, then $\gamma_{N,M}(P_N)$ and $\gamma_{N,M}(Q_N)$ are projections in $\A_M$. Furthermore, if $P_N \le Q_N$, then $\gamma_{N,M}(P_N) \le \gamma_{N,M}(Q_N)$.

Similarly to the inductive limit Hilbert space, we define an inductive limit algebra as follows. We first define a union algebra $\A_{\text{union}} \equiv \cup_{j=1}^{\infty} [\A_j]$, where $[\A_N]$ is an algebra of operators that act on $\h_{\text{union}}$ as $[O_N] [\ket{\psi}_M] = [\gamma_{N,M}(O_N) \ket{\psi}_M]$ if $M\ge N$ and as $[O_N] [\ket{\psi}_M] = [O_N \iota_{M,N}(\ket{\psi}_M)]$ if $M<N$. Then taking the double commutant of $\A_{\text{union}}$ in $\h$ defines a von Neumann algebra $\A \, = \, (\A_{\text{union}})'' \subset B(\h)$.

As an example, let us consider the same setup as before where $\h_N$ is a Hilbert space of $2N$ qubits and embedding maps between Hilbert spaces append the fixed state $\ket{\lambda}$ to the new qubits. Let us take $\A_N \subset B(\h_N)$ to be the algebra of operators acting non-trivially on the first qubit of each of the $N$ entangled pairs and let the embedding maps $\gamma_{N,M}$ append the identity $I$ to the new qubits. In this case, $\A_{\text{union}}$ is called a \textit{quasi-local algebra} as it acts non-trivially on finitely many qubits. Moreover, it is known in this setup that the von Neumann algebra $\A = (\A_{\text{union}})''$ is of Type III unless $\lambda \in \{0,1/2,1\}$ \cite{powers,araki1968classification}.  

Interestingly, we note that the values $\lambda \in \{0,1/2,1\}$ are precisely the values for which the state $\ket{\lambda}$ is a stabilizer state. Since a dense subspace of $\h$ consists of states whose tails are an infinite tensor product of $\ket{\lambda}$, we note that the states in this dense subspace have infinite magic if $\lambda \not\in \{0,1/2,1\}$. This suggests an intricate connection between infinite magic and Type III algebras. We make this connection rigorous in the next section.

\section{Main Results} \label{sec-main-thms}

In this section, we present the main contribution of this work and prove a connection between Type III von Neumann algebras and non-stabilizer states. To be precise, we consider an infinite-dimensional Hilbert space $\h$ that arises as a thermodynamic limit of a lattice shown in Fig.~\ref{fig-lattice}. In particular, $\h$ is the inductive limit of the sequence of Hilbert spaces $\h_1 \subset \h_2 \subset \cdots $ with embedding maps $\iota_{N,M}: \h_N \to\h_M$. Moreover, we consider a von Neumann algebra $\A\subset B(\h)$ that corresponds to the observables localized in the `right' region $R$ shown in Fig.~\ref{fig-lattice}. This algebra is also the inductive limit of the sequence $\A_1 \subset\A_2 \subset \cdots$, with embedding maps $\gamma_{N,M}:\A_N \to\A_M$ which are compatible with the embedding maps $\iota_{N,M}$; see Eq.~\eqref{eq-map-compatibility}.

As we reviewed in Sec.~\ref{sec-inductive-limit}, the maps $\iota_{N,M}$ that specify how a smaller Hilbert space embeds into a larger Hilbert space in the inductive sequence determine the macroscopic properties of the states that live in $\h$ and the bounded operators that are allowed in $\A$. Let us consider embedding maps under which the magic in states remains bounded from above. In this case, our goal is to prove that the algebra cannot be Type III. 

Our approach in proving this will be to consider a sequence of normal state functionals on $\A$. As we reviewed in Sec.~\ref{sec-algebras}, the space of state functionals is weak-$^*$ compact, which means that there exists a subsequence of these state functionals that converges in the weak-$^*$ topology to some state functional $\phi$ on $\A$. However, the limiting state functional is not guaranteed to be a normal functional (see footnote~\ref{example-singular-functional} for an example.) Nevertheless, this functional has a unique normal part $\phi^{(n)}$ as in Eq.~\eqref{functional-decomposition}. The following lemma establishes some properties of this state functional and its normal part that we will need to show that $\A$ cannot be of Type III.

\begin{lemma}\label{main-lemma}
    Consider an inductive system of Hilbert spaces and von Neumann algebras with embedding maps $\iota_{N,M}: \h_N\to\h_M$ for $M>N$ and consider a sequence of states $\{\ket{\phi}_N \in \h_N\}$. This defines a sequence of state functionals $\phi_N(\cdot) \equiv [\mbox{}_N\bra{\phi}] \cdot [\ket{\phi}_N]$ on $\A$ such that a subsequence $\{\phi_{\N}\}$ converges in the weak-$^*$ topology to a state functional $\phi$ on $\A$.
    \begin{enumerate}
        \item If the sequence of states is such that the overlap of $[\ket{\phi}_N]$ with a fixed reference state $\ket{\psi}\in\h$ is bounded from below by $\delta > 0$ for all $N$, then the normal part of $\phi$, denoted by $\phi^{(n)}$, is non-zero.
        \item If the states $\ket{\phi}_{N} \in \h_N$ are stabilizer states, then the restriction of $\phi$ on $[\A_{\N}]$ defines a decreasing sequence of projections that converge in the SOT to a projection $P\in \A$.
        \item If both of the above assumptions are true, then the projection $P$ is non-zero and $\phi^{(n)}(P) = \phi^{(n)}(1) < \infty$. Furthermore, $\phi^{(n)}$ satisfies $\phi^{(n)}(O'O) = \phi^{(n)}(OO')$ for any $O,O' \in P\A P$.
    \end{enumerate} 
\end{lemma}
\begin{proof}
    See Appendix~\ref{appendix-proof}.
\end{proof}

As alluded to earlier, we consider embedding maps under which the magic in states remains bounded from above. In particular, if we take an arbitrary state in $\h_N$ for some $N$, then the magic in its image in $\h_M$ for $M>N$ is bounded from above by $\M$ for all $M\ge N$. We show that the assumption of bounded magic translates to the assumptions of Lemma~\ref{main-lemma}, which we then use to prove the following theorem. 

\begin{theorem}\label{thm-2}
    Consider an inductive system of Hilbert spaces and von Neumann algebras with embedding maps $\iota_{N,M}: \h_N\to\h_M$ for $M>N$ as above. For any given state $\ket{\psi}_1 \in \h_1$, let us quantify the magic in the states $\iota_{1,N}(\ket{\psi}_1) \in \h_N$ for $N\ge 1$ by the min-relative entropy of magic, $\M_N = \M_{\text{stab}}\left(\iota_{1,N}(\ket{\psi}_1)\right)$. If the sequence $\M_1, \M_2, \cdots$ is bounded from above by $\M < \infty$, then the von Neumann algebra $\A$ contains a non-zero finite projection. Consequently, $\A$ cannot be of Type III.
\end{theorem}
\begin{proof}
    Let us denote the states $\iota_{1,N}(\ket{\psi}_1) \in \h_N$ by $\ket{\psi}_N$ and the state $[\ket{\psi}_1] = \lim_{N\to\infty} \ket{\psi}_N$ by $\ket{\psi}$. Since the min-relative entropy of magic of the states $\ket{\psi}_N$ is bounded from above by $\M$ for all $N$, we deduce that there exist stabilizer states $\ket{\phi}_N \in \h_N$ such that $|\mbox{}_N\braket{\psi}{\phi}_N| \ge \delta = e^{-\M/2} > 0.$ Since the inclusion maps are compatible with the inner product, we get $|\braket{\psi[}{\phi}_N]| = |\mbox{}_N\braket{\psi}{\phi}_N| \ge \delta$.

    Thus, the sequence of states $\ket{\phi}_N \in \h_N$ satisfies the assumptions of Lemma~\ref{main-lemma}. Therefore, there exists a non-zero projection $P\in \A$ and a normal functional $\phi^{(n)}$ such that $\phi^{(n)}(P) = \phi^{(n)}(1)<\infty$ and that $\phi^{(n)}(OO') = \phi^{(n)}(O'O) $. If we assume that $\phi^{(n)}$ is faithful on $P\A P$, then $\phi^{(n)}$ would be a trace on $P\A P$, and Lemma~\ref{lemma-tracial-corner} would immediately imply that $P$ is a finite projection in $\A$. However, there is no justification for assuming $\phi^{(n)}$ to be a faithful functional on the subalgebra $P\A P$. Nevertheless, as we reviewed in Sec.~\ref{sec-algebras}, a normal positive functional $\phi^{(n)}$ is faithful on the subalgebra $P_{\phi^{(n)}} \A P_{\phi^{(n)}}$, where $P_{\phi^{(n)}}$ is the support projection of $\phi^{(n)}$. Therefore, our approach now is to show how Lemma~\ref{lemma-tracial-corner} implies that the support projection $P_{\phi^{(n)}}$ is finite.
    
    To do this, we observe that the fact $\phi^{(n)}(P) = \phi^{(n)}(1) < \infty$ and the definition of the support projection implies $\phi^{(n)}\left(P_{\phi^{(n)}}P P_{\phi^{(n)}}\right) = \phi^{(n)}\left(P_{\phi^{(n)}}\right) < \infty$. Furthermore, it follows from this that     $\phi^{(n)}\left(P_{\phi^{(n)}}(1-P)P_{\phi^{(n)}}\right) = 0 $. Since $P_{\phi^{(n)}}(1-P)P_{\phi^{(n)}}$ is a positive operator in $P_{\phi^{(n)}}\A P_{\phi^{(n)}}$, the faithfulness of $\phi^{(n)}$ on $P_{\phi^{(n)}}\A P_{\phi^{(n)}}$ implies $P_{\phi^{(n)}}(1-P)P_{\phi^{(n)}} = 0 $, or equivalently $P_{\phi^{(n)}} \le P$.

    A direct consequence of $P_{\phi^{(n)}} \le P$ is that $P_{\phi^{(n)}} \A P_{\phi^{(n)}} \subseteq P \A P$. Therefore, we deduce from above that $\phi^{(n)}(OO') = \phi^{(n)}(OO')$ for any  $O,O' \in P_{\phi^{(n)}} \A P_{\phi^{(n)}}$, and so $\phi^{(n)}$ is a trace on $P_{\phi^{(n)}} \A P_{\phi^{(n)}}$. Moreover, since $\phi^{(n)}(P_{\phi^{(n)}}) < \infty$, Lemma~\ref{lemma-tracial-corner} implies $P_{\phi^{(n)}} \in \A $ is a finite projection.

    We have thus shown that the von Neumann algebra $\A$ contains a non-zero finite projection, and hence, $\A$ cannot be of Type III.
\end{proof}

This theorem implies that Type III algebras are not possible in the inductive limit if the states in the dense subspace $\h_{\text{union}}\subset \h$ contain bounded magic. This is the main result of this paper.

\section{Discussion} \label{sec-discussion}

In this paper, we have developed a connection between a Type III von Neumann algebra acting on some Hilbert space and the amount of magic in the states of this Hilbert space. In particular, we have considered infinite-dimensional systems that arise as a thermodynamic limit of finite lattice models, and have used the formalism of inductive systems to define the limiting Hilbert space $\h$ and a von Neumann subalgebra of bounded operators $\A\subset B(\h)$. The states and the operators that survive the thermodynamic limit (inductive limit) are determined by the embedding maps that specify how a finite-dimensional space/algebra embeds inside $\h$ and $\A$, respectively. We have thus shown that a necessary condition for the von Neumann algebra $\A$ to be of Type III is that the states in a dense subspace $\h_{\text{union}} \subset \h$ contain an infinite amount of magic. 

We now discuss some implications and extensions of our work:

\textit{CFTs are magical: } The study of magic in the context of quantum field theories was initiated in Ref.~\cite{White:2020zoz}, which considered a quantum spin chain (specifically, the $\mathbb{Z}_3$ Potts model at its critical point) that is known to flow to a conformal field theory (CFT) in the IR. Through numerical calculations, it was observed that the magic in the ground state of this spin chain scales with the size of the model and is expected to diverge in the thermodynamic limit. This observation led to the conclusion that conformal field theories are inherently magical, as their ground states have infinite magic \cite{White:2020zoz}. 

We now argue that our result provides a theoretical explanation for this numerical observation. Let us consider a sequence of spin chain models indexed by the number of spins, $N$, and let $\ket{\Omega}_N$ and $\h_N$ be the ground state and the Hilbert space of the $N^{\text{th}}$ model, respectively. The inductive system that leads to the conformal field theory in the limit involves embedding maps that map the ground state of a smaller spin chain to that of a larger one. Therefore, if the magic in the ground states $\ket{\Omega}_N$ did not grow with $N$, our results would immediately imply that the algebra of operators localized to a subregion is not Type III. This would yield a contradiction, as the local operator algebras in quantum field theories are known to be Type III von Neumann algebras. Therefore, the unbounded growth of magic in the ground states of critical spin chains is a necessary condition for the emergence of Type III von Neumann algebras in the thermodynamic limit.

\textit{Non-local magic: } Our focus in this work has been on the magic of the overall state, which measures the total number of non-Clifford gates needed to prepare it. However, there is a finer notion of magic for multi-partite systems, called non-local magic \cite{backreaction-magical}. For instance, for a quantum state $\ket{\psi}_{AB}$ of a bipartite system $\h_A\otimes\h_B$, the non-local magic is defined as the magic of the state $U_{A}U_{B} \ket{\psi}_{AB}$, minimized over all possible unitary operators $U_{A}$ and $U_B$. These unitaries are not necessarily Clifford, and hence, non-local magic captures the non-Clifford gates that generate entanglement between systems $A$ and $B$. 

We now argue that our result can easily be generalized to non-local magic. Consider the lattice system introduced in Fig.~\ref{fig-lattice}, where each lattice $\Lambda_N$ is split into a right and a left part, $\Lambda_N = R_N \cup L_N$. Consequently, the Hilbert space $\h_N$ factorizes as $\h_N = \h_{R_N}\otimes \h_{L_N}$ for all $N$. We claim that if the stabilizer states $\ket{\phi}_N \in \h_N$ in Lemma~\ref{main-lemma}~(2) are replaced with states of the form $\ket{\tilde{\phi}}_N = U_{N}U'_N \ket{\phi}_N$ for any unitaries $U_N \in \A_N = B(\h_{R_N})$ and $U'_N \in \A'_N = B(\h_{L_N})$, the conclusion of Lemma~\ref{main-lemma}~(2) remains unchanged. This is because the reduced states of $\ket{\tilde{\phi}}_N$ on $\h_{R_N}$ remain projectors, exactly like the reduced states of stabilizer states $\ket{\phi}_N$. This suggests a stronger result: a necessary condition for the emergence of Type III algebras is infinite non-local magic.

\textit{Type III$_1$ vs Type III$_\lambda$ algebras: } Type III von Neumann algebras can be further classified into Type III$_1$ or Type III$_\lambda$ algebras for $0<\lambda<1$. Distinguishing between these subtypes is of great physical interest, as Type III$_1$ factors (algebras with trivial center) are relevant for quantum field theories \cite{Haag:1996hvx} and for mixing and ergodic systems \cite{Emch:1994jt,Marrakchi:2023fth,Ouseph:2023juq,Gesteau:2023rrx,Furuya:2023fei}. Moreover, Type III$_1$ factors naturally emerge in the large $N$ limit of the AdS/CFT correspondence and explain the emergence of time, the event horizon, and bulk spacetime \cite{Faulkner:2018faa,DeBoer:2019kdj,Leutheusser:2021frk,Leutheusser:2021qhd,Lashkari:2024lkt}, and are a prerequisite for the calculation of the gravitational entropy in the $1/N$ correction \cite{witten,Chandrasekaran:2022eqq}. Even though our result provides a necessary condition for the emergence of Type III algebras in the thermodynamic limit, it is agnostic to the specific subtype. This is because our approach relies on the Murray-von Neumann classification of von Neumann algebras in terms of projections, which cannot differentiate between subtypes of Type III algebras. We expect that methods based on asymptotic ratio sets \cite{araki1968classification} or Connes' classification using modular flows \cite{Connes1973} may lead to stronger results. 

\section*{Acknowledgement}
We would like to thank Chris Akers, Oliver DeWolfe, Nima Lashkari, and Johannes Reinking for useful discussion and feedback. This work was supported by the Heising-Simons Foundation under Grant 2024-4848.

\appendix

\section{Proof of Lemma~\ref{lemma-tracial-corner}} \label{appendix-proof-trace}

\begin{proof}
    
To prove this lemma, let us consider a projection $Q \in \A$ such that $Q\le P$ and $Q \sim P$. Then we will show that under the assumptions of the lemma, the only possible $Q$ satisfying these two properties is $Q=P$, and hence, $P\in \A$ is a finite projection.

First, we recall that $Q\sim P$ implies that there exists a partial isometry $V\in\A$ such that $P = V^\dagger V$ and $Q = VV^\dagger$. Next, we note that the condition $Q \le P$ implies that $Q \in P\A P$. Furthermore, it implies that $V \in P\A P$. To see this, we write
\begin{align}
    0 = (1-P)P(1-P) = (1-P)V^\dagger V(1-P) = \left( V(1-P) \right)^\dagger \left(V(1-P)\right) \, ,
\end{align}
and hence, $V(1-P) = 0$. Repeating the same reasoning for $0 = (1-P)Q(1-P)$ yields $(1-P)V = 0$. By combining these two relations, we get $V = PVP$, which implies $V \in P\A P$.

Now since $V\in P\A P$ and since $\tau$ is a trace on $P\A P$, we get
\begin{align}
    \tau(P) = \tau(V^\dagger V) = \tau(V V^\dagger) = \tau(Q) \, .
\end{align}
Moreover, by assumption, $\tau(P) < \infty$. Therefore, we get
\begin{align}
\tau(P-Q) = \tau(P)-\tau(Q) = 0 \, .
\end{align}
Since $Q \le P$, we know that $P-Q$ is a positive operator in $P\A P$. Since $\tau$ is faithful on $P\A P$, $\tau(P-Q) = 0$ implies $P-Q = 0$.

Therefore, the only $Q \in \A$ such that $Q\in P$ and $Q\sim P$ is $Q=P$. Therefore, $P\in \A$ is a finite projection.
\end{proof}

\section{Proof of Lemma~\ref{main-lemma}} \label{appendix-proof}

\begin{proof}

Consider a subsequence of functionals $\{\phi_{\N}\}$ that converges to a functional $\phi$ on $\A$ in the weak-$^*$ topology. In particular, for any $O\in \A$, we have 
    \begin{align}
        \phi(O) = \lim_{\N\to\infty} \phi_{\N}(O) = \lim_{\N\to\infty} [\mbox{}_{\N}\bra{\phi}] O [\ket{\phi}_{\N}] \, .\label{eq-weak-star-limit}
    \end{align}
As reviewed in Sec.~\ref{sec-algebras}, this functional can be uniquely decomposed into a positive normal part $\phi^{(n)}$ and a positive singular part $\phi^{(s)}$. That is, $\phi = \phi^{(n)} + \phi^{(s)}$.

\textit{Part $1$}: Let us define a normal functional $\psi(\cdot) = \bra{\psi} \cdot \ket{\psi}$ where $\ket{\psi} \in \h$ is the reference state. We consider the \textit{trace distance} between state functionals $\phi$ and $\psi$, which is defined as $\Vert \phi - \psi\Vert = \sup_{O\in\A ; \Vert O\Vert \le 1} |\phi(O) - \psi(O)|$. Using Eq.~\eqref{eq-weak-star-limit}, we get $|\phi(O) - \psi(O)| = \lim_{\N \to\infty} |\phi_{\N}(O) - \psi(O)|$. Moreover, for any fixed $\N$, we have
\begin{align}
    |\phi_{\N}(O) - \psi(O)| &= \Big| [\mbox{}_{\N}\bra{\phi}] O [\ket{\phi}_{\N}] - \bra{\psi}O\ket{\psi}\Big| \, ,\nonumber\\
    &\le \Big\Vert  [\ket{\phi}_{\N}\mbox{}_{\N}\bra{\phi}]  - \ket{\psi}\bra{\psi}\Big\Vert_{1} \, \Vert O \Vert \, ,\nonumber\\
    &\le 2 \sqrt{1 - |\bra{\psi}[\ket{\phi}_N]|^2} \, \Vert O \Vert \, ,\nonumber\\
    &\le 2 \sqrt{1 - \delta^2 } \, \Vert O \Vert \, ,
\end{align}
where we have used the assumption that the overlap of $[\ket{\phi}_N]$ with $\ket{\psi}$ is bounded from below by $\delta > 0$ in the last step. Since this is true for all $\N$, we deduce that $|\phi(O) - \psi(O)| \le 2 \sqrt{1 - \delta^2 } \, \Vert O \Vert$, and hence, $\Vert \phi - \psi\Vert \le 2 \sqrt{1 - \delta^2 } < 2 $.

Now we prove $\phi^{(n)} \ne 0$ by contradiction. Let us assume $\phi = \phi^{(s)}$ is a singular state functional, and let us define $\omega = \phi - \psi$ and decompose it as $\omega = \omega^{(n)} + \omega^{(s)}$. By uniqueness, we deduce that $\omega^{(n)} = - \psi$ and $\omega^{(s)} = \phi$. Moreover, $\Vert \omega \Vert = \Vert \omega^{(n)} \Vert + \Vert \omega^{(s)} \Vert = \Vert -\psi \Vert + \Vert \phi \Vert = 2$. This is a contradiction since we have already shown $\Vert \phi - \psi \Vert < 2$.

\textit{Part $2$}: Let us now consider the restriction of $\phi$ on $[\A_{\N}]$. Since $[\A_{\N}]$ is a finite-dimensional algebra and $\phi$ is a state functional, we can represent the restriction of $\phi$ to $[\A_{\N}]$ in terms of a reduced density matrix $[\rho_{\N}] \in [\A_{\N}]$ and a canonical trace on $\A_{\N}$. That is, $\phi([O_{\N}]) = \text{tr}_{[\A_{\N}]}([\rho_{\N}] [O_{\N}])$ for all $[O_{\N}] \in [\A_{\N}]$. To derive $[\rho_{\N}]$, we note that for any $[O_{\N}] \in [\A_{\N}]$, Eq.~\eqref{eq-weak-star-limit} yields
    \begin{align}
        \phi([O_{\N}]) \, =& \lim_{\N'\to\infty} [\mbox{}_{\N'}\bra{\phi}] [O_{\N}] [\ket{\phi}_{\N'}] \, ,\\
        =& \lim_{\N'\to\infty} \mbox{}_{\N'}\bra{\phi} \gamma_{\N,\N'}\left(O_{\N}\right) \ket{\phi}_{\N'} \, ,\nonumber\\
        =& \lim_{\N'\to\infty} \text{tr}_{\gamma_{\N,\N'}(\A_{\N})}\left( \sigma^{(\N')}_{\N} \, \gamma_{\N,\N'}\left(O_{\N}\right) \right) \, ,\nonumber
    \end{align}
    where $\sigma_{\N}^{(\N')}$ is the density matrix of $\ket{\phi}_{\N'}$ on $\gamma_{\N,\N'}(\A_{\N})$. However, since $\sigma_{\N}^{(\N')} \in \gamma_{\N,\N'}(\A_{\N})$ and since embedding maps are injective, it means that there must exist some $\rho_{\N}^{(\N')} \in \A_{\N}$ such that $\sigma_{\N}^{(\N')} = \gamma_{\N,\N'}(\rho_{\N}^{(\N')})$. As a result, we get
    \begin{align}
        \phi([O_{\N}]) \, =& \, \lim_{\N' \to\infty} \text{tr}_{\A_{\N}}\left( \rho^{(\N')}_{\N} \, O_{\N} \right) \, \nonumber\\
        =&\, \text{tr}_{\A_{\N}}\left( \rho_{\N} \, O_{\N} \right) \, = \text{tr}_{[\A_{\N}]}\left( [\rho_{\N}] \, [O_{\N}] \right) \, ,\label{eq-subsequence-identity}
    \end{align}
    where $\rho_{\N} = \lim_{\N' \to\infty} \rho^{(\N')}_{\N} $. 
    
    We now argue that $[\rho_{\N}] \in [\A_{\N}]$ is proportional to a projection. Since $\ket{\phi}_{\N'}$ is a stabilizer state, Lemma~\ref{lemma-stab-proj} implies that $\sigma_{\N}^{(\N')}$ is proportional to a projection operator. Moreover, since the embedding maps $\gamma_{\N,\N'}$ are injective $*$-homomorphisms, we deduce that $\rho_{\N}^{(\N')} \in \A_{\N}$ must also be proportional to a projection operator for all $\N' > N$. This means that $\chi(\rho_{\N}^{(\N')}) \equiv S_{3}(\rho_{\N}^{(\N')}) - S_{2}(\rho_{\N}^{(\N')}) = 0$ for all $\N' > \N$, where $S_n(\rho)$ denotes the $n^{\text{th}}$ R\'enyi entropy of $\rho$. For finite-dimensional algebras like $\A_{\N}$, the R\'enyi entropies and hence $\chi$ is a continuous function of state, which means $\chi(\rho_{\N}) = \lim_{\N'\to\infty} \chi(\rho_{\N}^{(\N')}) = 0$. Therefore, $\rho_{\N}$ has a flat spectrum which implies it is proportional to a projection operator. Consequently, $[\rho_{\N}] \in [\A_{\N}]$ is also proportional to a projection, which we denote by $[P_{\N}]$.

    Therefore, the restriction of $\phi$ to $[\A_{\N}]$ defines a projection $[P_{\N}] \in [\A_N] $. Since this is true for all $\N$, we get a sequence of projections. Moreover, since the algebras $\{[\A_{\N}]\}$ are nested, the projections $\{[P_{\N}]\}$ form a decreasing sequence.

    It is known that a decreasing sequence of projections in a von Neumann algebra $\A$ converges in SOT to a projection in $\A$ \cite[Thm.~4.1.2]{murphy}. We denote this projection by $P$.
    
    \textit{Part $3$}: Now we combine parts $1$ and $2$. We observe from Eq.~\eqref{eq-subsequence-identity} that $\phi([P_{\N}]) = 1$ for all $\N$. This means that $\phi(1-[P_{\N}]) = \phi^{(n)}(1-[P_{\N}]) + \phi^{(s)}(1-[P_{\N}]) = 0$. Since both $\phi^{(n)}$ and $\phi^{(s)}$ are positive functionals, and since $1-[P_{\N}]$ is a projection, we deduce that $\phi^{(n)}(1-[P_{\N}]) = 0$. 

    Since $\{ 1 - P_{\N} \}$ forms an increasing sequence of projections and $\phi^{(n)}$ is a normal state, we get $\phi^{(n)}(1-P) = \lim_{\N\to\infty} \phi^{(n)}(1-P_{\N}) = 0$. Therefore, $\phi^{(n)}(P) = \phi^{(n)}(1) = \Vert \phi^{(n)}\Vert \ne 0$, and hence, $P \ne 0$. This proves that $P$ is non-zero and that  $\phi^{(n)}(P) = \phi^{(n)}(1) < \infty$. 

    Finally, we need to show that $\phi^{(n)}$ satisfies $\phi^{(n)}(OO') = \phi^{(n)}(O'O)$ for any $O,O' \in P\A P$. Since any $O$ can be written as a linear combination of at most four unitary operators, it is enough to prove that $\phi^{(n)}(U O) = \phi^{(n)}(OU ) $ for any $O \in P\A P$ and any unitary $U \in P\A P$. We will now prove this simpler condition.

    Let us consider a fixed $\N$ and let $U_{\N}$ be any unitary operator in $[P_{\N} \A_{\N} P_{\N}]$. Using this unitary, we define a transformed state functional $\phi_{U_{\N}}(\cdot) = \phi(U_{\N}^\dagger \cdot \, U_{\N})$. From Eq.~\eqref{eq-subsequence-identity}, we note that for any $[O_{\N}] \in [\A_{\N}]$, we get  
    \begin{align}
        \phi(U_{\N}^\dagger [O_{\N}] U_{\N}) =& \text{tr}_{[\A_{\N}]}\left( [\rho_{\N}] \, U_{\N}^\dagger [O_{\N}] U_{\N} \right) \\=& \text{tr}_{[\A_{\N}]}\left( [\rho_{\N}] \, [O_{\N}] \right) = \phi([O_{\N}]) \, ,\nonumber
    \end{align}
    where we have used the fact that $[\rho_{\N}]$ is proportional to the projector $[P_{\N}]$, and hence, commutes with $U_{\N} \in [P_{\N} \A_{\N} P_{\N}]$. Therefore, $\phi_{U_{\N}} = \phi$ on the subalgebra $[P_{\N} \A_{\N} P_{\N}]$. 
    
    Note that the normal part of $\phi_{U_{\N}}$ is given by $\phi^{(n)}_{U_{\N}} = \phi^{(n)}(U_{\N}^\dagger \cdot \, U_{\N})$. Since the decomposition of a state functional into a normal and a singular part as in Eq.~\eqref{functional-decomposition} is unique, and since both $\phi$ and $\phi_{U_{\N}}$ are equal on subalgebra $[P_{\N} \A_{\N} P_{\N}]$, we deduce that their normal parts must also be equal on this subalgebra. Therefore, we get
    \begin{align}
        \phi^{(n)}(U_{\N}^\dagger [O_{\N}] U_{\N}) = \phi^{(n)}([O_{\N}]) \, ,
    \end{align}
    or equivalently, 
    \begin{align}
        \phi^{(n)}(U_{\N} [O_{\N}] ) = \phi^{(n)}([O_{\N}]U_{\N}) \, ,\label{eq-cyclic-phin}
    \end{align}
    for all $[O_{\N}] \in [P_{\N} \A_{\N} P_{\N}]$. 

    Now we consider an arbitrary $O$ and an arbitrary unitary $U$ in the subalgebra $P\A P$. Since $P\A P$ is von Neumann algebra, and hence, is SOT-closed, we can find sequences $[O_{\N}] \in [P_{\N} \A_{\N} P_{\N}]$ and $U_{\N} \in [P_{\N} \A_{\N} P_{\N}]$ that converge in SOT to $O$ and $U$ in $P\A P$, respectively. Moreover, since $U_{\N}$ is unitary, the product $U_{\N}[O_{\N}]$ also converges to $UO$ in SOT:
    \begin{align}
         \lim_{\N} \Vert (U_{\N}[O_{\N}] - UO) \ket{\xi}\Vert =& \lim_{\N} \Vert (U_{\N}[O_{\N}] - U_{\N}O + U_{\N}O - UO) \ket{\xi}\Vert  \\ \le& \lim_{\N} \Vert U_{\N}([O_{\N}] - O)\ket{\xi}\Vert + \lim_{\N} \Vert (U_{\N} - U) O\ket{\xi}\Vert \nonumber\\
         =& \lim_{\N} \Vert ([O_{\N}] - O)\ket{\xi}\Vert + \lim_{\N} \Vert (U_{\N} - U) \ket{O\xi}\Vert = 0 \nonumber \, ,
    \end{align}
    where the first term in the last step vanishes because of $[O_{\N}] \to O$ in SOT whereas the second term vanishes because of $U_{\N} \to U$ in SOT. Since $\phi^{(n)}$ is a normal functional, we get
    \begin{align}
        \phi^{(n)}(UO) =& \lim_{\N} \phi^{(n)}(U_{\N}[O_{\N}]) \nonumber\\ =& \lim_{\N} \phi^{(n)}([O_{\N}]U_{\N}) = \phi^{(n)}(OU) \, ,
    \end{align}
    where we have used Eq.~\eqref{eq-cyclic-phin} in the intermediate step. This proves that $\phi^{(n)}(U O) = \phi^{(n)}(OU ) $ for any $O \in P\A P$ and any unitary $U \in P\A P$.
    
\end{proof}

\bibliographystyle{utcaps}
\bibliography{ref}
\end{document}